\documentclass{llncs}
\usepackage[utf8]{inputenc}
\usepackage{enumitem}

\makeatletter
\newcommand{\@chapapp}{\relax}%
\makeatother

\usepackage{amsfonts}
\usepackage{amsmath}

\usepackage{amsthm}
\usepackage{amssymb}
\usepackage{graphicx}
\usepackage{microtype}
\usepackage{wrapfig}
\usepackage[export]{adjustbox}
\usepackage{subfig}
\usepackage{array}
\usepackage{tabularx}
\usepackage{xparse}

\usepackage{algorithm}
\usepackage[noend]{algpseudocode}
\floatname{algorithm}{Procedure}

\newcommand{\tuple}[1]{\langle #1 \rangle}

\newcommand{\ceil}[1]{\left \lceil #1 \right \rceil }

\newcommand{\impl}{\Rightarrow}

\usepackage{letltxmacro}
\LetLtxMacro{\originaleqref}{\eqref}
\renewcommand{\eqref}{Eq.~\originaleqref}

\usepackage[hypertexnames=false]{hyperref}
\usepackage[%
    style=numeric-comp,sorting=nyt,
    sortcites=true,doi=false,url=false,isbn=false,
    giveninits=true,hyperref,backend=bibtex,bibencoding=ascii,maxbibnames=99]{biblatex}
\title{Linear-time algorithm for vertex 2-coloring without monochromatic triangles on planar graphs}
\titlerunning{Linear-time algorithm for vertex 2-coloring without monochromatic triangles on planar graphs}
\author{Micha\l~Karpi\'nski, Krzysztof~Piecuch}
\authorrunning{Micha\l~Karpi\'nski, Krzysztof~Piecuch}
\date{}
\institute{Institute of Computer Science, University of Wroc\l aw \\
Joliot-Curie 15, 50-383 Wroc\l aw, Poland \\
\email{\{karp,kpiecuch\}@cs.uni.wroc.pl}}

\begin{document}

  \maketitle

  \begin{abstract}
    In the problem of 2-coloring without monochromatic triangles (or triangle-tree 2-coloring), vertices of the simple, connected, undirected graph
    are colored with either {\em black} or {\em white} such that there are no 3 mutually adjacent vertices of the same color.
    In this paper we are positively answering the question posed in our previous work, namely, if there exists an algorithm
    solving 2-coloring without monochromatic triangles on planar graphs with linear-time complexity.
  \end{abstract}

  \section{Introduction}

  Vertex coloring problem and its variations have been studied extensively for many decades. In the classic coloring problem
  we are tasked to find the coloring of vertices such that there are no two adjacent vertices of the same color. Theoretical and practical applications of graph
  coloring spark interest in the research community to this day. One of the variations revolves around the idea of forbidding monochromatic subgraphs.
  In this paper we focus on monochromatic triangles, i.e., three mutually adjacent vertices of the same color. The motivation comes from the field of economy,
  namely, in the study of consumption behavior \cite{deb2008efficient}.

  In our previous paper \cite{karpinski2018vertex} we proposed a new structural parameter $\chi_3(G)$ which is the minimum
  number of colors needed to label the vertices of an undirected, simple graph $G$, such that there are no monochromatic triangles.
  We showed efficient, linear-time algorithms for finding $\chi_3(G)$ (and triangle-free $\chi_3(G)$-coloring) for several classes of graphs,
  including outerplanar graphs, chordal graphs and graphs with bounded maximum degree $\Delta$, with $\Delta \leq 4$.
  Our approach was based on the so called {\em standard recoloring strategy},
  where most of the work is relegated to finding the classic coloring with optimal $k=\chi(G)$ number of colors (where $\chi(G)$
  is the chromatic number). After that, we label $V_i$ to be the set of vertices colored $i$, for $1 \leq i \leq k$.
  Next, for each $0 \leq j \leq \ceil{k/2}-1$ we recolor sets $V_{2j+1} \cup V_{2j+2}$ with $j$ (we may need to add empty set $V_{k+1}$,
  if $k$ is odd). Since every $V_i$ is an independent set, then after recoloring, any monochromatic cycle is of even length.
  Therefore the resulting coloring is triangle-free. We note that similar technique can be found in the literature,
  for example, in the work of Stein \cite{stein1971b}. We can see the drawback of this framework: efficient algorithms for finding
  optimal coloring without monochromatic triangles depend on the existence of efficient algorithms for finding optimal (or near-optimal) classic coloring. 

  We say that the graph is {\em planar}, if we can embed it in the plane, i.e., it can be drawn on the plane in such a way that its edges intersect
  only at their endpoints. The application of standard recoloring strategy fails to provide a linear-time algorithm for finding 2-coloring
  without monochromatic triangles on planar graphs, as the best known algorithm for finding classic 4-coloring on planar graphs runs
  in quadratic time \cite{papadimitriou1981clique}. The situation does not improve even if the input graph is classically 3-colorable,
  as the best known algorithm for classically 4-coloring planar graphs with this property still runs in $O(n^2)$ time \cite{kawarabayashi2010simple}.

  In this paper we are closing the gap left in our previous work by presenting a linear-time algorithm for triangle-free 2-coloring on planar graphs
  that does not rely on the standard recoloring strategy.

  \subsection{Related work}

  In the field of graph coloring, the class of planar graphs has been of particular interest. Several interesting results that
  coincide with our research have emerged, for example, Angelini and Frati \cite{angelini2012acyclically}
  study planar graphs that admit an acyclic 3-coloring -- a proper coloring in which every 2-chromatic subgraph is acyclic.
  Another result is of Kaiser and \v{S}krekovski \cite{kaiser2004planar}, where they prove that every planar graph
  has a 2-coloring such that no cycle of length 3 or 4 is monochromatic. Thomassen \cite{thomassen20082}
  considers list-coloring of planar graphs without monochromatic triangles.

  Several hardness results for our problem also exist (for non-planar cases)
  and can be found in our previous papers \cite{karpinski2017vertex,karpinski2018vertex} and also in the work of Shitov \cite{shitov2017tractable}.
  It is worth noting that the hardness results for classic coloring provides additional motivation for our work. 
  We have efficient algorithms for finding $\chi_3$ on planar graphs and graphs with $\Delta=4$,
  but it is $\mathcal{NP}$-hard to find the $\chi$ even on 4-regular planar graphs \cite{dailey1980uniqueness}.
  This serves as an evidence, that the triangle-free coloring problem is easier than the classic coloring
  problem and we prove it in this paper for planar case.

  \subsection{Structure of the paper}

  In Section 2 we formalize the main problem and provide necessary definitions and notations. Section 3 contains the main contribution, i.e.,
  the linear-time algorithm for triangle-free 2-coloring on planar graphs. We leave concluding remarks in Section 4.

  \section{Preliminaries}

  In the previous section we mentioned two distinct types of colorings. Here we formalize their definitions.
  Let $G=(V,E)$ be a graph with vertex set $V$ and edge set $E$.
  A {\em classic k-coloring} of a simple graph is a function $c : V \rightarrow \{1,\dots, k\}$, such that there are no two
  adjacent vertices $u$ and $v$, for which $c(u)=c(v)$. Given $G$, the smallest $k$ for which there exists a classic $k$-coloring
  for $G$ is called the {\em chromatic number} and is denoted as $\chi(G)$. A {\em triangle-free k-coloring} of a simple graph is a function
  $c : V \rightarrow \{1,\dots, k\}$, such that there are no three mutually adjacent vertices $u$, $v$ and $w$, for which $c(u)=c(v)=c(w)$.
  If such vertices exist, then the induced subgraph ($K_3$) is called a {\em monochromatic} triangle.
  Given $G$, the smallest $k$ for which there exists a triangle-free $k$-coloring for $G$ we call the {\em triangle-free chromatic number} and
  we denote it as $\chi_3(G)$. If $k=2$ we define $\overline{c(v)}=1$, if $c(v)=2$, and $\overline{c(v)}=2$, if $c(v)=1$.

  We use standard definitions and notations from graph theory. We remind the reader of the ones we use in this paper.
  A planar graph is {\em maximal}, if connecting any two vertices in it with a new edge
  makes it non-planar. This implies that every face of a maximal planar graph is a triangle.
  A triangle in the planar graph embedding that is not a face we call an {\em inner triangle}.
  We can make any planar graph to be maximal in linear-time \cite{read1986new,hagerup1991triangulating}.
  We assume from now on that we are solving the triangle-free coloring problem for maximal planar graphs only.
  We can do this, because every coloring (triangle-free coloring) of a maximal planar graph
  is also a coloring (triangle-free coloring) of the original planar graph.
  A {\em dual graph} of a planar graph $G$ is a graph that has a vertex for each face of $G$.
  The dual graph has an edge whenever two faces of $G$ are separated from each other by an edge.
  
  For any input planar graph presented in this paper, we assume that we are also given an embedding of such graph on a plane,
  where every segment is straight. It can be computed in linear time \cite{nishizeki1988planar},
  so it is safe to make such assumption.

  A graph $G$ is $k$-edge-connected if the minimum number of edges whose deletion disconnects $G$ is at least $k$.
  A graph $G$ is $k$-regular if every vertex in $G$ is of dergee $k$. A single edge
  whose removal disconnects a graph is called a {\em brigde}. A edge set $M \subset E$ is a matching,
  if every vertex from $V$ is incident to at most one edge in $M$. A matching $M$ is called {\em perfect} if every vertex is incident to some edge in $M$.
  An {\em augmenting path} for a matching $M$ is a path $\bar{e}$ with an odd number of edges $e_1,\dots,e_p$ such that
  it begins and ends at vertices which are not covered by $M$ and all edges from the set $\{e_2,e_4,\dots\}$ belong to $M$.
  The symmetric difference of $M$ and $\bar{e}$ yields a matching
  having one more edge than $M$. Augmenting paths are used in the maximum matching algorithms,
  for example in the Hungarian Algorithm \cite{kuhn1955hungarian}.

  For brevity, when graph $G=(V,E)$ is given, we assume that $|V|=n$ and $|E|=m$.

  \section{The Algorithm}

  Here we give the main result. We present the algorithm for triangle-free 2-coloring of planar graphs in three steps:

  \begin{enumerate}[label=\Roman*.]
    \item We show the algorithm, assuming that there are no inner triangles in the graph, then
    \item we show the algorithm with additional assumption, that the outer face of the graph is already colored, and finally
    \item we combine the previous results to find the triangle-free 2-coloring in any maximal planar graph in linear-time.
  \end{enumerate}

  \subsection{Algorithm I}

  Let $G=(V,E)$ be a simple, maximal planar graph, where $|V|>4$ (the other case is trivial), such that there are no inner triangles in $G$.
  Let $G'=(V',E')$ be a graph that is dual to $G$. Graph $G'$ has the following properties:

  \begin{enumerate}
    \item it is planar, because $G$ is planar,
    \item has no bridges, because $G$ has no loops,
    \item it is 3-regular, because every face in $G$ is a triangle,
    \item has no loops or multi-edges, because $G$ is $3$-edge-connected (because $G$ is simple, maximal and $|V|>4$).
  \end{enumerate}

  Petersen \cite{petersen1891theorie} proved that every 3-regular bridgeless graph has a perfect matching. Biedl et al. \cite{biedl2001efficient}
  showed how to find such matching in linear-time, assuming that the graph is planar. Graph $G'$ is planar, 3-regular and bridgeless, therefore
  we can use the aforementioned result to find a perfect matching $M$ in $G'$ in $O(n)$ time. From now on we assume that one such matching $M$ is already computed.

  \begin{definition}\label{def:bicolor}
    An edge $e \in E$, which is dual to the edge $e' \in E'$ such that $e' \not\in M$ is called {\em bicolor}.
  \end{definition}

  \begin{figure}[!t]
    \centering
    \includegraphics[scale=1.0]{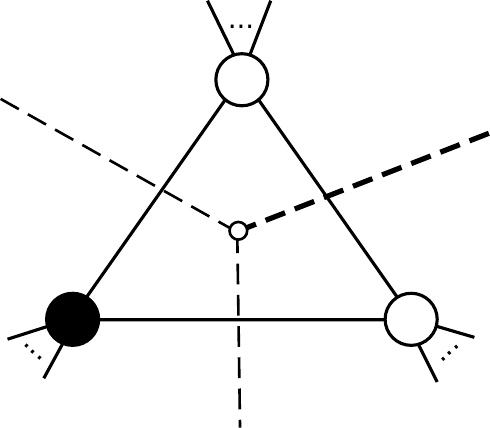}
    \caption{A face of some maximal planar graph and its dual edges. Coloring of the original graph is determined by the dual edge which belongs to the matching (bold dashed line).}
    \label{fig:ex1}
  \end{figure}

  \begin{example}
    The idea of the algorithm is to color vertices of the bicolor edges with two different colors. For example, consider a single face of some
    maximal planar graph as shown in Figure \ref{fig:ex1}. Dashed lines represent dual edges (edges from $G'$) and one of them (bold one) belongs
    to the perfect matching $M$. Notice that only one of the dual edges in the figure can be in $M$, therefore for this face (and any other face) there
    can be only one non-bicolor edge in the original graph, therefore two other edges in the face are bicolor and we can properly color the face
    (in triangle-free sense). Sample coloring is shown as a black/white fill.
  \end{example}

  \begin{definition}\label{def:consistent}
    Let $v$ be arbitrarily chosen from $V$. A coloring $c : V \rightarrow \{1,2\}$ is consistent with $M$ iff $c(v) = 1$, and for each $u\not=v$:

    \begin{itemize}
      \item $c(u) = 1$, if there is a simple path from $v$ to $u$ that contains even number of bicolor edges,
      \item $c(u) = 2$, if there is a simple path from $v$ to $u$ that contains odd number of bicolor edges.
    \end{itemize}
  \end{definition}

  \begin{lemma}
    Coloring $c$ (of $G$) that is consistent with $M$ is well-defined.
  \end{lemma}

  \begin{proof}
    Graph $G$ is maximal, therefore it is connected, which means that for each pair of vertices $u,v \in V$, there exist a simple path from $u$ to $v$.
    Assume, by contradiction, that there exist two different simple paths from $u$ to $v$, one of them containing odd number of bicolor edges and the
    other one containing even number of bicolor edges. Therefore there exist a simple cycle in $G$ that contains odd number of bicolor edges. Remove
    all edges of $M$ from $G'$. $G'$ becomes 2-regular, therefore it consists of only simple cycles. Each simple cycle in $G'$ crosses each simple cycle in $G$
    even number of times, i.e., each simple cycle in $G'$ contains even number of edges that are dual to some edges in every simple cycles of $G$.
    Therefore each simple cycle in $G$ contains even number of bicolor edges, a contradiction. Additionally, the choice of $v$ in Definition \ref{def:consistent}
    is irrelevant, as every triangle-free coloring where $c(v)=1$ can be made into a triangle-free coloring where $c(v)=2$ by simply flipping the color of all vertices.
  \end{proof}

  The algorithm for obtaining coloring that is consistent with $M$ is a simple {\em depth-first search} procedure, which we will call {\em Algorithm I}. 
  First, we choose unmarked vertex $v$ arbitrarily from $G$, and we mark it, and we set $c(v)=1$. Then for each unmarked neighbour $u$ of $v$ we determine
  if the $(u,v)$ is bicolor, and we color $u$ accordingly, i.e., if $(u,v)$ is bicolor then $c(u)=\overline{c(v)}$, otherwise $c(u)=c(v)$.
  Then, we recursively run the procedure on $u$. Algorithm stops when all vertices are marked. We now prove the correctness of this algorithm.

  \begin{lemma}\label{lma:correct1}
    Coloring $c$ (of $G$) that is consistent with $M$ is triangle-free.
  \end{lemma}

  \begin{proof}
    Assume by contradiction that under the coloring $c$ there exist a monochromatic triangle in $G$.
    This triangle is a face, since there are no inner triangles in $G$.
    This implies that no edge in the triangle is bicolor. Let $v'$ be the node representing this triangle in $G'$.
    By Definition \ref{def:bicolor} all edges incident to $v'$ are in $M$. This contradicts the fact that $M$ is a matching.
  \end{proof}

  \begin{theorem}\label{thm:1}
    Given a maximal planar graph $G=(V,E)$ which contains no inner triangles, we can find a triangle-free 2-coloring for $G$ in $O(n)$ time.
  \end{theorem}

  \begin{proof}
    By Lemma \ref{lma:correct1}, Algorithm I returns a triangle-free coloring, therefore we only need to prove the time complexity.
    Dual graph $G'$ can be easily computed in $O(n)$, assuming the embedding of $G$ is represented using DCEL (doubly connected edge list) data structure \cite{de1997computational}.
    Matching $M$ can be found in $O(n)$ \cite{biedl2001efficient}.
    The DFS procedure of Algorithm I runs in $O(n+m)$ time, assuming that we prepared the data structures such that: 1) checking if an edge belongs to $M$ takes constant time, which
    can be done by labeling each edge during the computation of $M$, and 2) determining the dual edge in $G'$ for some edge in $G$ takes constant time, which can
    be achieved by storing additional pointer with each edge in $G$ during the computation of $G'$.
    The fact that $G$ is planar -- and therefore $m \in O(n)$ -- concludes the proof.
  \end{proof}

  \subsection{Algorithm II}

  Let $G=(V,E)$ be the same as in Algorithm I, and assume that $u,v,w \in V$ are vertices of the outer face of $G$. Additionally, assume
  that $c(u)$, $c(v)$ and $c(w)$ is given, where $c$ is partial assignment of $\{1,2\}$ into $V$. We will now present an algorithm for extending
  $c$ to the entire set $V$, such that $c$ is triangle-free. Without loss of generality assume that $c(u)=c(v)=\overline{c(w)}$.

  Let $G'=(V',E')$ be dual to $G$ and let $e'=(u',v')$ be dual to $e=(u,v)$. If $e'$ belongs to some perfect matching $M$ of $G$, then $e$ is not bicolor,
  and therefore the other two edges of the outer face are bicolor, therefore Algorithm I can extend partial coloring $c$ to be consistent with $M$.
  Generalized Petersen's Theorem \cite{schonberger1934} guarantees that such $M$ exists. We show how to find such matching in linear time.

  \begin{lemma}\label{lma:correct2}
    Let $G$,$G'$,$M$,$c$,$u,v,w$,$e$ and $e'$ be as discussed above. Perfect matching $M'$ of $G$, for which $e' \in M'$, can be computed in $O(n)$ time.
  \end{lemma}

  \begin{proof}
    Find perfect matching $M'$ in $O(n)$ time \cite{biedl2001efficient}. If $e' \in M'$, then we are done. Assume that $e' \not\in M'$.
    Let $u^+$ and $v^+$ be a pair of new vertices. Construct a graph $G^+=(V^+,E^+)$,
    where $V^+=V' \cup \{u^+,v^+\}$ and $E^+ = E' \cup \{(u',u^+),(v',v^+)\} \setminus \{e'\}$. Graph $G^+$
    has the following properties:

    \begin{itemize}
      \item it has a perfect matching, for example, $M \cup \{(u',u^+),(v',v^+)\} \setminus \{e'\}$,
      \item $M'$ is a matching in $G^+$, because $e' \not\in M'$ and $e' \not\in E^+$,
      \item its perfect matching contains one more edge than $M'$, because $|V^+|=|V'|+2$,
      \item its perfect matching contains both $(u',u^+)$ and $(v',v^+)$, because those are the only
        edges incident to $u^+$ and $v^+$, respectively.
    \end{itemize}

    The above properties imply that there exists an augmenting path which extends $M'$ to a perfect
    matching $M^+$ of graph $G^+$. Goldberg and Tarjan \cite{goldberg1988new} shows how to compute such
    augmenting path in linear time. The perfect matching of $G'$ which contains $e'$ is then defined as
    $M^+ \setminus \{(u',u^+)$,$(v',v^+)\} \cup \{e'\}$. 
  \end{proof}

  \begin{theorem}\label{thm:2}
    Given a maximal planar graph $G=(V,E)$ which contains no inner triangles, for which the outer face $u,v,w$
    is colored such that $c(u)=c(v)=\overline{c(w)}$, we can find a triangle-free 2-coloring for $G$ in $O(n)$ time.
  \end{theorem}

  \begin{proof}
    Directly from Lemma \ref{lma:correct2} and by using Algorithm I.
  \end{proof}

  \subsection{Algorithm III}

  In this section we construct a linear-time algorithm for computing a triangle-free 2-coloring of a given maximal, planar graph $G$.
  First, we present the necessary definitions.

  \begin{figure}[!t]
    \centering
    \subfloat[~]{\label{fig:triangle1}\includegraphics[width=0.45\textwidth,valign=t]{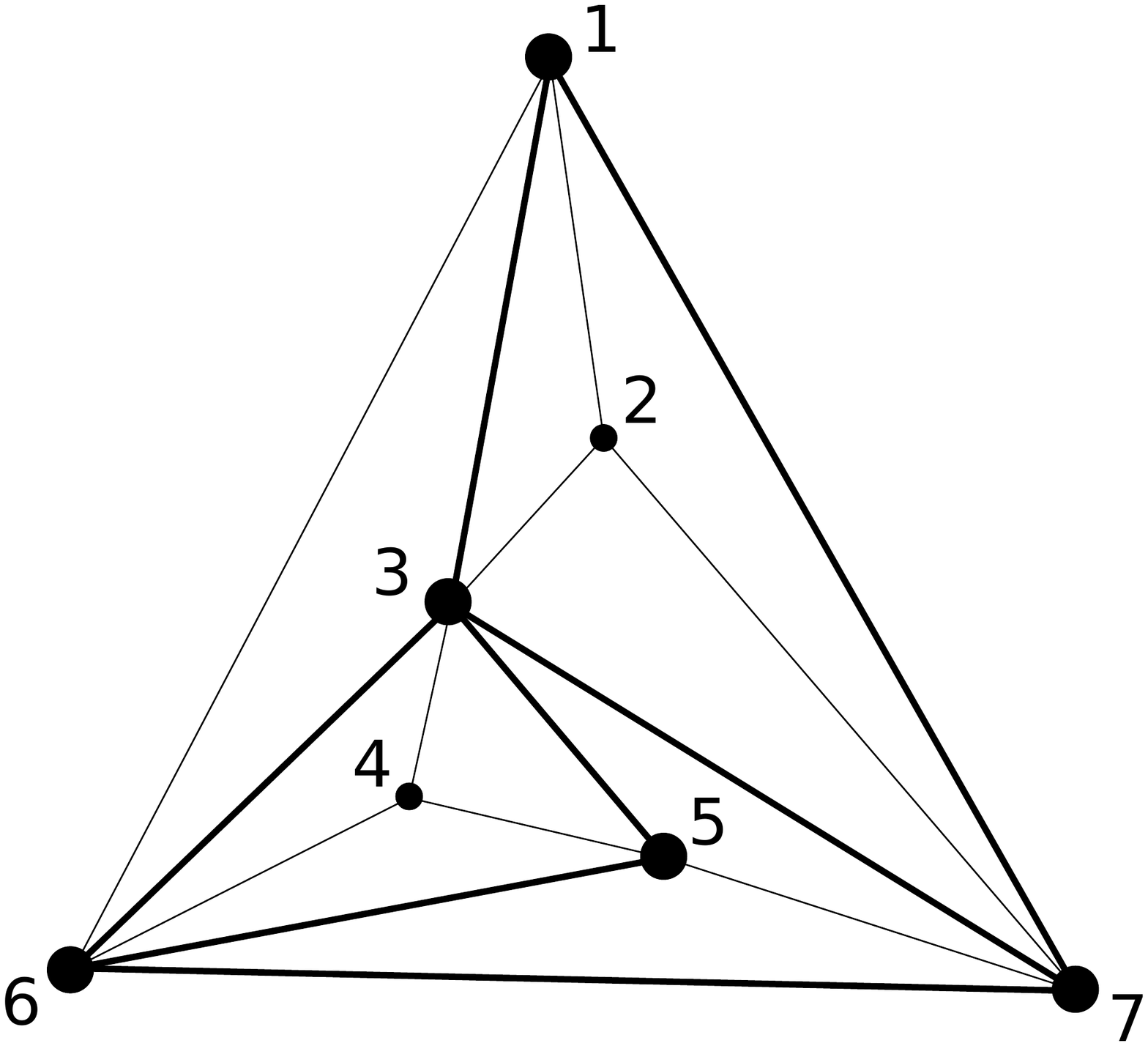}}~~
    \subfloat[~]{\label{fig:triangle2}\includegraphics[width=0.5\textwidth,valign=t]{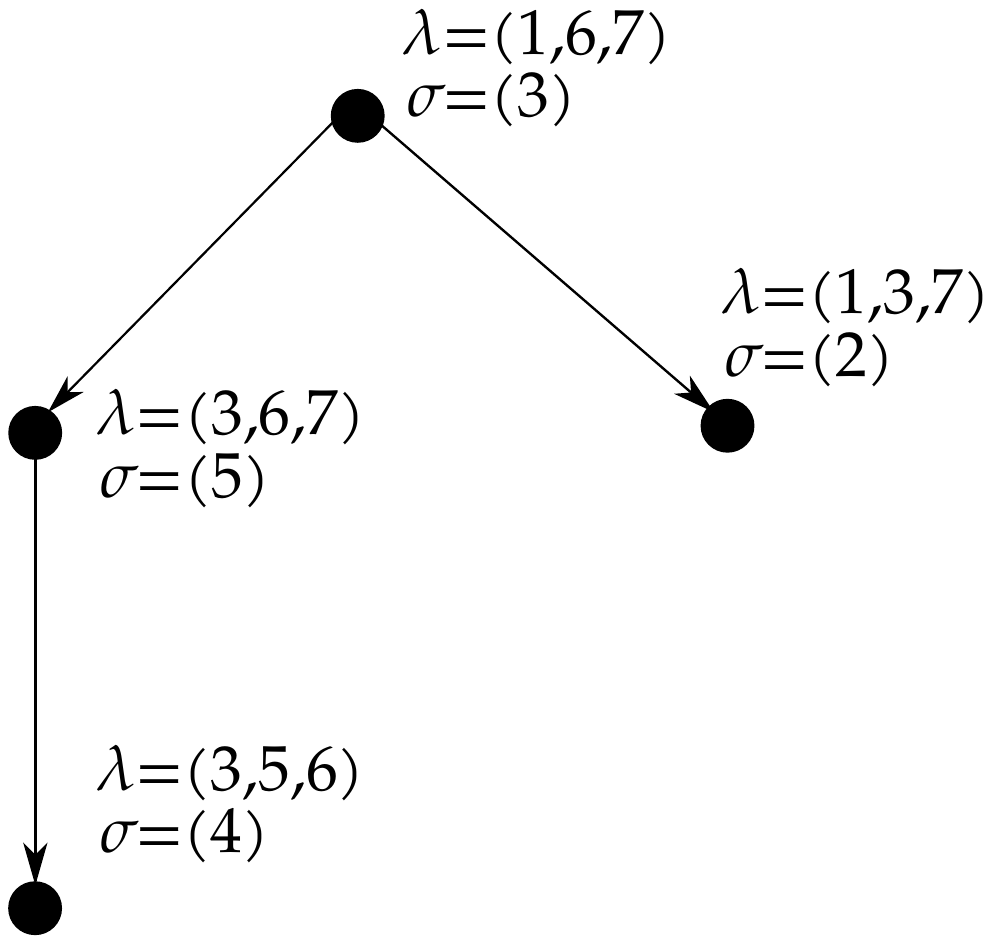}}
    \caption{A sample embedding of a maximal, planar graph with highlighted inner triangles. On the right we see the triangle hierarchy of the embedding.}
    \label{fig:triangle}
  \end{figure}

  \begin{definition}[pivotal triangles]
    A triangle in a maximal, planar graph embedding of $G$ is called {\em pivotal}, if it is an inner triangle in $G$ or an outer face of $G$.
  \end{definition}

  \begin{definition}[triangle inside]
    Let $t$ be a triangle in a maximal, planar graph $G$. The {\em inside} of $t$ is a plane fragment bounded by the edges
    of $t$ in the embedding of $G$. Edges and vertices of the triangle are not part of its inside. 
  \end{definition}

  \begin{definition}
    A triangle $t$ contains a triangle $s$ iff the inside of $s$ is a proper subset of the inside of $t$. We say that a triangle $t$
    contains a vertex $v$, if $v$ belongs to the inside of $t$.
  \end{definition}

  \begin{example}
    Let us consider the embedding presented in Figure \ref{fig:triangle1}. It consists of four pivotal triangles: the outer face
    $(1,6,7)$, and the three inner triangles $(1,3,7)$, $(3,6,7)$ and $(3,5,6)$. The triangle $(1,6,7)$ contains all others, and the triangle
    $(3,6,7)$ contains the triangle $(3,5,6)$. The triangle $(3,6,7)$ contains vertices $4$ and $5$, but does not contain the vertex $2$.
  \end{example}

  \begin{definition}
    Let $G$ be a maximal, planar graph, embedded on a plane, and let $A$ be the set of its pivotal triangles.
    We define a relation $\sqsubset$ on $A$, such that for any $t,s \in A$, $s \sqsubset t$ if and only if $t$ contains $s$.  
  \end{definition}

  By the abuse of notation we will write $v \sqsubset t$, where $v$ is a vertex, to say that a triangle $t$ contains $v$.
  Let $A$ be the set of triangles of the embedding of some maximal, planar graph. Observe that $(\sqsubset, A)$ is a partially ordered set.

  \begin{definition}[triangle hierarchy]
    Let $G$ be a maximal, planar graph, embedded on a plane, and let $A$ be the set of its pivotal triangles. Let $D=(A,F)$ be a directed graph,
    with vertices corresponding to the set of pivotal triangles $A$, and the set of edges $F$ defined as follows:
    for each $t,s \in A$, $\tuple{t,s} \in F$ iff $s \sqsubset t$ and there is no $r \in A \setminus \{t,s\}$ such that $s \sqsubset r \sqsubset t$.
    We call graph $D$ a {\em triangle hierarchy} of $G$.
  \end{definition}

  In any planar graph embedding, edges do not intersect (except at endpoints) and every triangle consists of three vertices, therefore any two triangle insides are either
  disjoint or one contains the other. Therefore a triangle hierarchy of $G$ is a tree rooted at the outer face of $G$.

  Let $D=(A,F)$ be a triangle hierarchy of some maximal, planar graph $G=(V,E)$. For each triangle $t \in A$ we define $\lambda(t)$ to be the set of vertices
  forming $t$, and $\sigma(t)=\{ v \in V \, : \, v \sqsubset t \wedge \forall_{s\in A}(s \sqsubset t \impl v \not \sqsubset s) \}$. Informally, $\sigma(t)$
  is the set of vertices, such that $t$ contains them, but no triangle lower in the hierarchy contains them.

  \begin{example}
    In Figure \ref{fig:triangle2} the triangle hierarchy of the embedding presented in Figure \ref{fig:triangle1} is given, together with the values of $\lambda$
    and $\sigma$ for each triangle in the hierarchy.
  \end{example}

  The outline of the algorithm is as follows. We would like to reduce the problem of triangle-free coloring to the triangle-free coloring without inner triangles,
  in order to use Algorithms I and II as blackboxes. Notice, that if we remove all the vertices in the insides of all inner triangles (in some embedding), we are left with
  the embedding where all triangles are faces. After this operation, we can use Algorithm I to color the graph. Next, each face (that is not the outer face) together with
  its inner vertices that has been removed earlier, induce a maximal, planar graph embedding where the outer face is already colored. Therefore we can run the algorithm
  recursively, now using Algorithm II as a blackbox. This translates to a top-to-bottom recursive algorithm on the triangle hierarchy,
  where vertices involved in the computation at each level of recursion are exactly those in $\lambda$ and $\sigma$ lists of that level in the triangle hierarchy.
  At the root, the Algorithm I is used to color the vertices, and for all the recursive calls we use Algorithm II.
  The algorithm terminates when all leafs have been processed, and therefore the entire graph has been colored.

  We now prove the necessary properties of this strategy in order to prove its correctness (Lemmas \ref{lma:h1} and \ref{lma:h2}).
  Next, we show that the size of the triangle hierarchy is linear (Lemmas \ref{lma:l1} and \ref{lma:l2}),
  and therefore we prove that the algorithm runs in a linear time. 
  In the following lemmas we assume that $D=(A,F)$ is a triangle hierarchy of some maximal, planar graph $G=(V,E)$, where $|V|>4$.
  The assumption on $|V|$ guarantees that $D$ is not empty.

  \begin{lemma}\label{lma:h1}
    Let $s \in A$, such that $s$ is not the root in $D$, and let $v \in \sigma(s)$.
    For each $t \in D$, which is an ancestor of $s$, $v \not\in \lambda(t) \cup \sigma(t)$.
  \end{lemma}

  \begin{proof}
    Let $t \in D$ be an ancestor of $s$. There exists at least one $v \in \sigma(s)$, because $t$ is pivotal, and not a root,
    therefore it is an inner triangle. Take any such $v \in \sigma(s)$. Assume by contradiction that $v \in \lambda(t) \cup \sigma(t)$.
    If $v \in \lambda(t)$, then $v$ is one of the vertices forming $t$, but because $s \sqsubset t$, we also have $v \sqsubset t$, a contradiction.
    If $v \in \sigma(t)$, then by the definition of $\sigma$: $v \sqsubset t \wedge \forall_{s\in A}(s \sqsubset t \impl v \not \sqsubset s)$, and this implies
    that $v \not \sqsubset s$, a contradiction.
  \end{proof}

  \begin{lemma}\label{lma:h2}
    For each $t \in A$, a graph embedding $G_t$ induced by the vertices $\lambda(t) \cup \sigma(t)$ is maximal, planar and has no inner triangles.
  \end{lemma}

  \begin{proof}
    $G_t$ is maximal and planar, because $G$ is maximal and planar and because vertices from $\lambda(t)$ form a triangle (the outer face of $G_t$).
    Assume by contradiction, that $G_t$ contains an inner triangle $s=(u,v,w)$.
    Therefore there exists a vertex $z$, such that $z \sqsubset s$ and $z \in \sigma(s)$.
    Thus, by Lemma \ref{lma:h1}, $z \not\in \lambda(t) \cup \sigma(t)$.
    Therefore $z$ is not one of the vertices of $G_t$, a contradiction.
  \end{proof}

  \begin{lemma}\label{lma:l1}
    $|A| \leq |V|-3$.
  \end{lemma}

  \begin{proof}
    By Lemma \ref{lma:h1}, for each vertex $v \in V$ that is not on the outer face there exists at most one $t \in A$, such that $v \in \sigma(t)$.
    Therefore $\sum_{t\in A} |\sigma(t)| \leq |V|-3$. Additionally, each triangle $t \in A$ is pivotal, therefore $|\sigma(t)|>0$, which completes the proof.    
  \end{proof}

  \begin{lemma}\label{lma:l2}
    $\sum_{t \in A} |\lambda(t) \cup \sigma(t)| \leq 4|V| - 12$.
  \end{lemma}

  \begin{proof}
    By Lemma \ref{lma:l1}, $|A| \leq |V|-3$. For each triangle $t \in A$, $|\lambda(t)|=3$, therefore $\sum_{t\in A} |\lambda(t)| \leq 3|V|-9$.
    From the proof of Lemma \ref{lma:l1} we know that $\sum_{t\in A} |\sigma(t)| \leq |V|-3$. Combining those two bounds completes the proof.
  \end{proof}

  \begin{algorithm}[t!]
    \caption{Algorithm III - recursive procedure}\label{alg:3}
    \begin{algorithmic}[1]
      \Require {Maximal, planar graph $G=(V,E)$, embedded on a plane, and its triangle hierarchy $D=(A,F)$, where $t \in A$ is the root.
        The partial coloring $c$ of $G$, where only vertices from $\lambda(t) \cup \sigma(t)$ are colored.}
      \Ensure{Triangle-free 2-coloring of $G$.}
      \If {$|V| \leq 4$}
        \Return $c$
      \EndIf
      \ForAll {$\tuple{t,s} \in F$}
        \State Let $G_s$ be the graph embedding induced by $\lambda(s) \cup \sigma(s)$.
        \State Run Algorithm II on $G_s$ and $c$, and let $c_s$ be the resulting coloring.
        \State Recursively run this procedure on $G'_s$ induced by $s$ and its inside vertices, 
          the triangle hierarchy rooted at $s$, and the partial coloring of its outer face set by $c_s$.
          Let $c'_s$ be the coloring returned by this recursive call.
      \EndFor
      \State \Return A combined coloring of $c$ and all $c'_s$'s, for each $\tuple{t,s} \in F$.
    \end{algorithmic}
  \end{algorithm}

  We now combine every result presented so far into the final algorithm, which we will simply call Algorithm III.
  Let $G=(V,E)$ be a maximal, planar graph embedded on a plane. Do the following:

  \begin{enumerate}
    \item Construct a triangle hierarchy $D=(A,F)$ of $G$.
    \item Let $t \in A$ be the root of $D$. Run Algorithm I on $\lambda(t) \cup \sigma(t)$ and let $c$ be the resulting coloring.
    \item Run Procedure \ref{alg:3} on $G$, $D$ and $c$. Return the resulting coloring.
  \end{enumerate}

  \begin{theorem}
    Given a maximal planar graph $G=(V,E)$ we can find a triangle-free 2-coloring for $G$ in $O(n)$ time.
  \end{theorem}

  \begin{proof}
    Use Algorithm III to color $G$. We now prove that the resulting coloring is indeed triangle-free. By Lemma \ref{lma:h2}, the graph induced by $\lambda(t) \cup \sigma(t)$
    is maximal, planar, and has no inner triangles. Therefore by Theorem \ref{thm:1}, the coloring $c$ is triangle-free. We now prove that Step 3 returns a triangle-free
    coloring. We do this by induction on the height $h$ of $D$. Assume again that $|V|>4$, thus $D$ is not empty (the other case is trivial).

    If $h=0$, then $D$ has only one element, which means that there are no inner triangles in $G$,
    therefore all vertices of $G$ are already colored before running Procedure \ref{alg:3}.
    Take any $h>0$ and assume that for all $h'<h$, if $D$ is of height $h'$, then
    Step 3 returns a triangle-free coloring. Let $D$ be of height $h$.
    Take any $\tuple{t,s} \in F$. The graph $G_s$ is maximal, planar, and has no inner triangles due to Lemma \ref{lma:h2}.
    Thus, by Theorem \ref{thm:2}, $c_s$ is triangle-free, therefore by induction hypothesis $c'_s$ is triangle-free.
    Since $c$ is triangle-free, extending it with all $c'_s$'s (for each $\tuple{t,s} \in F$) produces a triangle-free coloring of $G$.
    Finally, the only way we color vertices in Algorithm III is by running Algorithms I and II, therefore the returned coloring is a 2-coloring.

    We now prove the time complexity. A triangle hierarchy can be constructed in $O(n)$ time,
    by first computing all pivotal triangles in $O(n)$ time \cite{chiba1985arboricity}, and then traversing the graph using a DFS method on a DCEL structure,
    starting from a vertex of the outer face and remembering which pivotal triangle we are currently in.
    By Theorem \ref{thm:1}, Step 2 runs in $O(n)$ time. The running time of Step 3 is also bounded by $O(n)$, due to
    Theorem \ref{thm:2} and Lemmas \ref{lma:l1} and \ref{lma:l2}.
  \end{proof}

  \section{Conclusions}

  In this paper we presented a linear-time algorithm for finding a triangle-free 2-coloring of any planar graph.
  Thus, we have positively answered one of the questions posed in our previous paper \cite{karpinski2018vertex}.
  This result proves that the triangle-free coloring problem is easier than the classic coloring problem, and
  can possibly shed the light on the latter's complexity, which interests the graph theory community to this day. 

  For further research, there are still open problems to be solved. For the positive result one can look for an algorithm
  that finds $\chi_3(G)$ in graphs with $\Delta(G) \geq 5$ (where $\Delta(G)$ is the largest vertex degree of $G$).
  For the negative side, one can look for the smallest $\Delta(G)$, for which the
  triangle-free $k$-coloring problem is $\mathcal{NP}$-hard.

  {\printbibliography}

@article{karpinski2017vertex,
  title={Vertex 2-coloring without monochromatic cycles of fixed size is NP-complete},
  author={Karpi{\'n}ski, Micha{\l}},
  journal={Theoretical Computer Science},
  volume={659},
  pages={88--94},
  year={2017},
  publisher={Elsevier}
}

@inproceedings{karpinski2018vertex,
  title={On vertex coloring without monochromatic triangles},
  author={Karpi{\'n}ski, Micha{\l} and Piecuch, Krzysztof},
  booktitle={International Computer Science Symposium in Russia},
  pages={220--231},
  year={2018},
  organization={Springer}
}

@article{shitov2017tractable,
  title={A tractable NP-completeness proof for the two-coloring without monochromatic cycles of fixed length},
  author={Shitov, Yaroslav},
  journal={Theoretical Computer Science},
  volume={674},
  pages={116--118},
  year={2017},
  publisher={Elsevier}
}

@book{read1986new,
  title={A new method for drawing a planar graph given the cyclic order of the edges at each vertex},
  author={Read, Ronald Cedric},
  year={1986},
  publisher={Faculty of Mathematics, University of Waterloo}
}

@article{hagerup1991triangulating,
  title={Triangulating a planar graph},
  author={Hagerup, T and Uhrig, C},
  journal={Library of Efficient Datatypes and Algorithms (LEDA), software package, Max-Plank Instit{\"u}t f{\"u}r Informatik, Saarbr{\"u}cken},
  year={1991}
}

@article{petersen1891theorie,
  title={Die Theorie der regul{\"a}ren graphs},
  author={Petersen, Julius},
  journal={Acta Mathematica},
  volume={15},
  number={1},
  pages={193--220},
  year={1891},
  publisher={Springer}
}

@article{biedl2001efficient,
  title={Efficient algorithms for Petersen's matching theorem},
  author={Biedl, Therese C and Bose, Prosenjit and Demaine, Erik D and Lubiw, Anna},
  journal={Journal of Algorithms},
  volume={38},
  number={1},
  pages={110--134},
  year={2001},
  publisher={Elsevier}
}

@article{schonberger1934,
  author={Sch{\"o}nberger, Tibor},
  title={Ein Beweis des Petersenschen Graphensatzes},
  journal={Acta Scienta Mathematica Szeged},
  year={1934},
  pages={51-57},
  volume={7}
}

@article{goldberg1988new,
  title={A new approach to the maximum-flow problem},
  author={Goldberg, Andrew V and Tarjan, Robert E},
  journal={Journal of the ACM (JACM)},
  volume={35},
  number={4},
  pages={921--940},
  year={1988},
  publisher={ACM}
}

@article{chiba1985arboricity,
  title={Arboricity and subgraph listing algorithms},
  author={Chiba, Norishige and Nishizeki, Takao},
  journal={SIAM Journal on computing},
  volume={14},
  number={1},
  pages={210--223},
  year={1985},
  publisher={SIAM}
}

@article{deb2008efficient,
  title={An efficient nonparametric test of the collective household model},
  author={Deb, Rahul},
  journal={Available at SSRN 1107246},
  year={2008}
}

@article{papadimitriou1981clique,
  title={The clique problem for planar graphs},
  author={Papadimitriou, Christos H and Yannakakis, Mihalis},
  journal={Information Processing Letters},
  volume={13},
  number={4-5},
  pages={131--133},
  year={1981},
  publisher={Elsevier}
}

@article{kawarabayashi2010simple,
  title={A simple algorithm for 4-coloring 3-colorable planar graphs},
  author={Kawarabayashi, Ken-ichi and Ozeki, Kenta},
  journal={Theoretical Computer Science},
  volume={411},
  number={26-28},
  pages={2619--2622},
  year={2010},
  publisher={Elsevier}
}

@article{angelini2012acyclically,
  title={Acyclically 3-colorable planar graphs},
  author={Angelini, Patrizio and Frati, Fabrizio},
  journal={Journal of combinatorial optimization},
  volume={24},
  number={2},
  pages={116--130},
  year={2012},
  publisher={Springer}
}

@article{kaiser2004planar,
  title={Planar graph colorings without short monochromatic cycles},
  author={Kaiser, Tom{\'a}{\v{s}} and {\v{S}}krekovski, Riste},
  journal={Journal of Graph Theory},
  volume={46},
  number={1},
  pages={25--38},
  year={2004},
  publisher={Wiley Online Library}
}

@article{thomassen20082,
  title={2-list-coloring planar graphs without monochromatic triangles},
  author={Thomassen, Carsten},
  journal={Journal of Combinatorial Theory, Series B},
  volume={98},
  number={6},
  pages={1337--1348},
  year={2008},
  publisher={Elsevier}
}

@article{dailey1980uniqueness,
  title={Uniqueness of colorability and colorability of planar 4-regular graphs are NP-complete},
  author={Dailey, David P},
  journal={Discrete Mathematics},
  volume={30},
  number={3},
  pages={289--293},
  year={1980},
  publisher={Elsevier}
}

@incollection{de1997computational,
  title={Computational geometry},
  author={De Berg, Mark and Van Kreveld, Marc and Overmars, Mark and Schwarzkopf, Otfried},
  booktitle={Computational geometry},
  pages={1--17},
  year={1997},
  publisher={Springer}
}

@article{kuhn1955hungarian,
  title={The Hungarian method for the assignment problem},
  author={Kuhn, Harold W},
  journal={Naval research logistics quarterly},
  volume={2},
  number={1-2},
  pages={83--97},
  year={1955},
  publisher={Wiley Online Library}
}

@book{nishizeki1988planar,
  title={Planar graphs: Theory and algorithms},
  author={Nishizeki, Takao and Chiba, Norishige},
  volume={32},
  year={1988},
  publisher={Elsevier}
}

@article{stein1971b,
  title={B-sets and planar maps},
  author={Stein, Sherman},
  journal={Pacific Journal of Mathematics},
  volume={37},
  number={1},
  pages={217--224},
  year={1971},
  publisher={Mathematical Sciences Publishers}
}

\end{document}